\documentclass[a4paper,12pt]{article}
\usepackage{listings}
\usepackage{anysize}
\marginsize{2cm}{2cm}{2cm}{2cm}
\usepackage{amsmath, amsthm, amssymb, amsfonts} %%%for math environment
\usepackage{mathtools}
\usepackage{graphicx}
\usepackage{lscape}
\usepackage{tikz}
\usepackage{caption,subcaption}
\usetikzlibrary{shapes,arrows}
\usepackage{float}
\usetikzlibrary{decorations.pathmorphing} % noisy shapes
\usetikzlibrary{fit}					% fitting shapes to coordinates
\usetikzlibrary{backgrounds}
\usepackage{setspace}
\restylefloat{table}
\usepackage{epstopdf}
\usepackage{xcolor}
\usepackage{amsmath}
\usepackage{amssymb}
\usepackage{algorithm}
\usepackage{booktabs}
\usepackage{multirow}
\usepackage{algpseudocode}

\newtheorem{theorem}{Theorem}

\newtheorem{proposition}[theorem]{Proposition}
\newtheorem{remark}[theorem]{Remark}

\makeatother

\title{Shortermism and excessive risk taking in optimal execution with a target performance}

\author{Emilio Barucci\thanks{Department of Mathematics, Politecnico di Milano}\and Yuheng Lan\thanks{Department of Mathematics, Politecnico di Milano}
}

\begin{document}
	%\nocite{*}
	\global\long\def\sgn{\mathop{\mathrm{sgn}}}
	\linespread{1.5}
	
	\maketitle
	
	\abstract{We deal with the optimal execution problem when the broker's goal is to reach a performance barrier avoiding a downside barrier. The performance is provided by the wealth accumulated by trading in the market, the shares detained by the broker evaluated at the market price plus a slippage cost yielding a quadratic inventory cost. Over a short horizon, this type of remuneration leads, at the same time, to a more aggressive and less risky strategy compared to the classical one, and over a long horizon the performance turns to be poorer and more dispersed.}
	
	\vspace{1cm}
	
	\begin{flushleft}
	\end{flushleft}
	
	\newpage
	\section{Introduction}
	Since the seminal contribution of \cite{ALGCH}, a large literature has examined the optimal execution problem: how to sell a stock of shares to maximize wealth and, at the same time, minimize inventory costs over a finite horizon; see \cite{CAR15,DON} for a review of the literature.  
	In this paper, we deal with optimal execution when the broker aims to reach an upper performance barrier avoiding the downside barrier. Performance is provided by the wealth accumulated by trading in the market, the market value of the shares held by the broker, and a quadratic inventory cost.
	
	The paper is related to the literature on optimal contracts for order execution between the owner of the shares (principal) and the broker (agent) who effectively trades on the market, for example, see \cite{BALD,LARS}. We can interpret our setting as a remuneration scheme consisting of a positive fixed remuneration if the broker reaches the upper performance barrier and a null remuneration if the downside barrier is reached; see \cite{BR95,BR98,BR99} for the analysis of this type of remuneration scheme in the asset management setting.
	The remuneration scheme leads the broker to maximize the probability of success (reaching the upper performance barrier) avoiding ruin (reaching the downside performance barrier). 
	
	A remuneration based on achieving a performance target is considered suspicious because it is a non-linear scheme and therefore may lead to excessive risk taking and shortermism, as shown in the asset management/executive compensation literature, see \cite{BARMA2,BARMA,BASA,CARP,GRIN,ROSS}. 
	In this paper, we investigate the effect of the above remuneration scheme on the optimal execution strategy. 
	
	To our knowledge, there is little research on optimal execution with a target on performance. The only exception being provided by \cite{JAKIN}, where the optimal acquisition problem is analyzed with a price limiter: the broker minimizes the cost of acquisition of a certain number of shares with a cost penalizing the inventory and a limit price threshold (when the asset price touches the upper price limit the broker acquires all the remaining shares at the market price impacted by the number of traded shares). Notice that in the above paper a barrier on the asset price is considered, instead in our paper we consider symmetric barriers on the performance.  
	
	We show two main results.
	
	First of all, the optimal strategy with a performance target foresees to liquidate the shares at a much higher rate compared to the solution obtained solving the classical optimal execution problem over a finite horizon, i.e., exponential rather than linear rate. This renders a higher performance compared to the classical solution in the short run because of smaller inventory costs, but not in the long run because trades occur with high execution costs and the broker is sooner trapped with almost no shares to sell. Over a short horizon, the optimal strategy yields a higher performance on average and a smaller dispersion, compared to the solution with a finite horizon, and over a long horizon the reverse is observed.
	
	The strategy turns out to be aggressive and conservative at the same time. Shares are quickly liquidated, performance goes up, but then the broker balances the goal of reaching the upper barrier avoiding the lower one. As a consequence, there is a high probability of reaching the upper barrier and of remaining between the two barriers with a very low probability of reaching the lower barrier. Instead, the strategy obtained by solving the classical problem allocates a significant probability to reach the lower barrier. 
	
	The key insight is that the strategy is aggressive compared to the classical solution, but balances the probability of touching the upper barrier against the probability of touching the lower barrier, which is almost null. Instead, the classical solution takes risks on both sides.
	Although the strategy obtained with a performance target takes care of the downside risk, we show that it leads to good performance only over a short period and not over a longer horizon.
	Therefore, we can conclude that a remuneration scheme based on reaching a performance target avoiding the downside barrier is affected by shortermism but doesn't entail excessive risk taking.
	
	The paper is organized as follows. In Section \ref{MOD} we present the model and formulate our problems.  
	In Section \ref{P1} we address the maximization problem.
	In Section \ref{NUM} we develop some numerical analysis/simulations.
	In Section \ref{PRIC} we compare the strategy with the one with a price limiter. Section \ref{CONC} concludes.

	\section{The Model}
	\label{MOD}
	We consider the classical optimal execution problem in continuous time, see \cite{ALGCH,CAR15,DON}. 
	The broker holds $Q_0$ shares in $t=0$ and wants to sell them in the market.
	
	Denote $v(t)$ the number of shares sold by the broker at time $t$. The amount of shares of the 
	broker evolves according to the following law of motion: 
	\begin{equation}
		\label{dQ}
		dQ(t)=-v(t)dt, \ \ Q(0)=Q_0.
	\end{equation}
	
	The market price of the asset is permanently affected by quantity sold in the market by the broker:
	\begin{equation}
		\label{PRICE}
		dS(t)=-f(v(t))dt+ \sigma dW(t), 
	\end{equation}
	where $W$ is a standard Brownian motion, and $f(\cdot)$ models the permanent price impact of the shares sold by the broker.
	
	We distinguish between the asset market price $S(t)$ and the price at which the order is executed $\hat{S}(t)$ (execution price) which is impacted by the trading of the broker:
	\begin{equation}
		\label{IMPACT}
		\hat{S}(t)=S(t)-g(v(t)),
	\end{equation}
	where $g(\cdot)$ represents the temporary price impact.
	
	The wealth of the broker evolves as follows
	$$
	dX(t)=\hat{S}(t)v(t) dt, \ \ X(0)=0.
	$$
	
	In what follows, we stick to the simplest model, which foresees a linear impact of the broker's trades on both the market price and the execution price: $f(v)=bv , \ g(v)=lv$.

	We consider the following performance criterion:
	$$
	Y(t)=X(t) + Q(t)(S(t)-\gamma Q(t)),
	$$
	which includes the wealth accumulated by trading on the market ($X(t)$), the shares detained by the broker evaluated at the market price ($Q(t)S(t)$) minus a quadratic cost associated with the unsold number of shares $(\gamma Q^2(t))$. 
	$\gamma Q(t)$ can be interpreted as the slippage cost for selling rapidly $Q(t)$ shares in the market (the execution price is $S(t)-\gamma Q(t)$).
	
	The broker is risk neutral and the classical optimal execution problem concerns the maximization of the expected value of the performance by the terminal date $T$:
	\begin{equation}
		\mathcal{J}^{v}(t, x, q, s)=\mathbb{E}[X(T)+ Q(T)(S(T)-\gamma Q(T))|X(t)=x, Q(t)=q, S(t)=s].
	\end{equation}
	
	We refer to this as Problem P0, see \cite{CAR15,DON}. The optimal trading strategy and the optimal inventory are 
	\begin{equation}
		v^*(t)=\frac{b-2\gamma}{2l+(2\gamma-b)(T-t)}
		Q^*(t)
	\end{equation}
	\begin{equation}
		Q^*(t)=\frac{2l+(2 \gamma -b)(T-t)}{2l+(2\gamma- b)T}Q_0
	\end{equation}
	and the value function is 
	\begin{equation}
		\mathcal{J}(t, x, q, s)=x+qs+h_2(t)q^2
	\end{equation}
	where
	$$
	h_2(t)=\big(\frac{1}{2 l}(T-t)+\frac{1}{-2\gamma+b})^{-1}-\frac{1}{2}b.
	$$
	Notice that the inventory decreases linearly over time.
	
	In what follows, we assume that the broker manages the inventory with a target performance. 
	Given a trading strategy $v$, we define
	$\tau^v_a$ as 
	$$
	\tau^v_a=\inf \{t>0: Y(t)=a\}.
	$$
	The broker maximizes the probability of reaching a high performance target (success) avoiding a low performance target (ruin). Given $k, \ h$ and the trading strategy $v$, we set  $\tau^{v}=min(\tau^{v}_k, \tau^{v}_h)$, the goal of the broker is to maximize the probability of reaching $h$ before $k$:
	\begin{equation}
		\label{MAX}
		\sup_{v\in\mathcal{A}}\mathcal{J}^{v}(y)=Prob (\tau^v=\tau^v_h |Y(0)=y),
	\end{equation}
	where $k<y<h$ and $\mathcal{A}$ is the admissible set of trading strategies:
	$$
	\mathcal{A}=\{v(t)|v\text{ is non-negative and uniformly bounded from above}\}.
	$$
	We refer to this as Problem P1 and its value function as
	$\mathcal{J}(y)$.
	
	Adopting $Y(t)$ as a performance criterion leads the broker to take care of both the cash obtained from the trade $(X(t))$ and the market value of the inventory at time $t$, which takes into account a slippage cost to trade all the shares immediately at that time ($Q(t)S(t)-\gamma Q^2(t)$). We can decompose this second component as the value of the shares at the market price ($Q(t)S(t)$) and a quadratic inventory cost ($\gamma Q^2(t)$).
	Therefore, the reward is not purely market based (either cash obtained by trading or mark-to-market of the residual inventory), the broker also faces an inventory cost. This implies that the broker has to find the right balance between quickly liquidating the shares and holding them with a penalization higher than the temporary price impact ($\gamma$ is much higher than $l$). 
	
	In summary, there are three components in the broker's performance: cash, market value of the residual inventory, and quadratic inventory cost.

	\section{Maximize probability of success/minimize the probability of ruin}
	\label{P1}
	
	Our problem is to find the optimal strategy that maximizes (\ref{MAX}).
	The optimal strategy is obtained in the following proposition.
	
	\begin{proposition}
		\label{Prop1}
		The optimal strategy for Problem P1 is
		\begin{equation}
			\label{OPT}
			v^{*}(t)= \frac{2\gamma-b}{2l}Q^*(t), \ \ 
			v^*(t)=\frac{2\gamma-b}{2l}e^{\frac{b-2\gamma}{2l}t}Q_0
		\end{equation}   
		\begin{equation}   
			\label{OPT1}
			Q^{*}(t)=e^{\frac{b-2\gamma}{2l}t}Q_0.
		\end{equation}
		The value function is 
		\begin{equation}
			\mathcal{J}(y)=\frac{e^{-\lambda y}-e^{-\lambda k}}{e^{-\lambda h}-e^{-\lambda k}},
		\end{equation}
		where
		$$\lambda=\frac{(2\gamma-b)^2}{2l\sigma^2}.$$
	\end{proposition}
	\begin{proof}
		By standard methods in stochastic control, see \cite[chapter VI]{FLE}), the value function $\mathcal{J}(y, q)$ satisfies the second order HJB (Hamilton-Jacobi-Bellman) equation
		\begin{align}
			\left\{
			\begin{aligned}
				&\sup_{v\in\mathcal{A}}\mathcal{L}^{v}\mathcal{J}(y, q)=0, &&(y, q)\in\Omega, \\
				&\mathcal{J}(y, q)=\phi(y, q), &&(y, q)\in\partial^{*}\Omega,
				\label{HJB-linear case-Target ax+bqs-gammaq2a}
			\end{aligned}
			\right.
		\end{align}
		where $\mathcal{L}^{v}$ is the infinitesimal generator operator
		\begin{align*}
			\mathcal{L}^{v}\mathcal{J}(y, q)=\frac{1}{2}&\sigma^2q^2\partial_{yy}\mathcal{J}
			+\{[- lv^2- qbv+2\gamma qv]\partial_{y}\mathcal{J}-v\partial_{q}\mathcal{J}\}
		\end{align*}
		and
		\[
		\phi(y, q) = \left\{
		\begin{aligned}
			&1 &&y=h \\
			&0 &&y=k.
		\end{aligned}
		\right.
		\]
		Let us assume that the HJB equation admits a classical solution such that $\partial_{yy}\mathcal{J}<0$ and $\partial_{y}\mathcal{J}>0$. 
		If $-b q\partial_{y}\mathcal{J}+2\gamma q\partial_{y}\mathcal{J}-\partial_{q}\mathcal{J}\leq0$, then $\sup$ in (\ref{HJB-linear case-Target ax+bqs-gammaq2a}) is obtained for $v= 0$. To get a non-trivial solution we assume $-b q\partial_{y}\mathcal{J}+2\gamma q\partial_{y}\mathcal{J}-\partial_{q}\mathcal{J}>0$, then the supremum in (\ref{HJB-linear case-Target ax+bqs-gammaq2a}) is attained at 
		$$
		v^*(y, q)=\frac{-b q\partial_{y}\mathcal{J}+2\gamma q\partial_{y}\mathcal{J}-\partial_{q}\mathcal{J}}{2 l\partial_{y}\mathcal{J}}
		$$
		and the value function $\mathcal{J}^{v}(y, q)$ satisfies 
		\begin{align}
			\left\{
			\begin{aligned}
				&\frac{1}{2}\sigma^2q^2\partial_{yy}\mathcal{J}+\frac{(-b q\partial_{y}\mathcal{J}+2\gamma q\partial_{y}\mathcal{J}-\partial_{q}\mathcal{J})^2}{4l\partial_{y}\mathcal{J}}=0 &&(y, q)\in\Omega \\
				&1 &&y= h \\
				&0 &&y= k.
			\end{aligned}
			\right.
		\end{align}
		Observing the boundary conditions, we assume that the value function depends only on $y$ and, therefore, the HJB equation becomes:
		\begin{align*}
			&\frac{1}{2}\sigma^2q^2\partial_{yy}\mathcal{J}+\frac{(-b q\partial_{y}\mathcal{J}+2\gamma q\partial_{y}\mathcal{J})^2}{4l\partial_{y}\mathcal{J}}=0,
		\end{align*}
		yielding a linear second-order ordinary differential equation:
		\begin{align*}
			&\frac{1}{2}\sigma^2\partial_{yy}\mathcal{J}+\frac{(-b+2\gamma )^2\partial_{y}\mathcal{J}}{4l}=0,
		\end{align*}
		with boundary conditions
		$$
		\mathcal{J}(h)=1,\mathcal{J}(k)=0.
		$$
		We assume that the value function satisfies the ansatz:  
		$$
		\mathcal{J}(y) = C e^{-\lambda y} + D.
		$$
		
		We can compute
		$$
		C=\frac{1}{e^{-\lambda h}-e^{-\lambda k}}, D=\frac{-e^{-\lambda k}}{e^{-\lambda h}-e^{-\lambda k}}.
		$$
		Substituting the value function back into the expression of $v^{*}$, we get
		$$
		v^{*}= \frac{2\gamma-b}{2l}q.
		$$
		Integrating $dQ^*(u) = -v^*(u) du$ over $[0, t]$, we obtain the optimal inventory and the corresponding optimal strategy
		\begin{equation}
			\int_{0}^{t}\frac{dQ^*(u)}{Q^*(u)}=\int_{0}^{t}\frac{(b-2\gamma)}{2l}du \rightarrow Q^{*}(t)=e^{\frac{b-2\gamma}{2l}t}Q_0.
		\end{equation}
		$v^{*}$ satisfies the HJB equation, the process $Y^{*}(t)$ 
		for $v^{*}(t)$ and $Q^{*}(t)$
		has constant coefficients, and the value function $\mathcal{J}$ is sufficiently smooth, satisfying the Lipschitz condition. Consequently, all conditions of the classical verification theorem are met, see \cite[Theorem VI.4.2]{FLE}.
	\end{proof}
	
	Note that the amount of shares for Problem P1 decreases exponentially while it decreases linearly for Problem P0. 
	
	The solution shows some similarities with the one that is obtained for the same problem in the asset management setting, see \cite{BR95}.
	The optimal trading strategy doesn't depend on the two barriers $h$ and $k$. Moreover, 
	the optimal trading rate maximizes the drift of $Y(t)$. As a matter of fact, 
	by Ito's formula we get that the drift of $dY(t)$ is a quadratic function of $v(t)$ reaching the maximum for  $v(t)=\frac{2\gamma-b}{2l}Q(t)$. Therefore, selling the fixed fraction $\frac{2\gamma-b}{2l}$ of outstanding shares renders the highest expected increase in the performance. This result agrees with what is observed for the same problem in the asset management, in that case the optimal portfolio is provided by the golden rule, i.e., the portfolio that maximizes the expected logarithmic return of wealth. 
	
	As in the classical optimal execution problem, we assume $2 \gamma -b>0$, i.e., twice the slippage cost be higher than the permanent effect of the broker's order. 
	
	As expected, the constant fraction of the inventory to be liquidated in the market increases in the slippage cost (inventory cost) $\gamma$ and decreases in the permanent and temporary effects of the trade ($b$ and $l$). The trading strategy is not affected by the asset price volatility.

	\begin{remark}
		If $Q(t)=0$ then it is optimal to set $v(s)=0$ for $s \ge t$ and therefore $Y(s)=Y(t)$.
	\end{remark}

	\section{Numerical and Simulation analysis}
	\label{NUM}
	In what follows, we provide a numerical analysis of the solution for Problem P1. The baseline set of parameters comes from \cite{CAR15} and is reported in Table \ref{parame0}.\footnote{Without loss of generality we set $Q_0=Y_0=1$, this implies $S_0=1+\gamma$.} 
	As an illustration we consider a $\pm 5\%$ barrier on $Y_0$, later on we will investigate what happens changing the two barriers.

	\begin{table}[htbp]
		\centering
		\caption{Baseline set of parameters of the model.}
		\label{parame0}
		\begin{tabular}{|c|c|c|}
			\hline
			Parameter & Value & Interpretation \\
			\hline
			$b$ & 0.001 & permanent impact \\
			$l$ & 0.001 & temporary impact \\
			$\gamma$ & 0.1 & slippage cost \\
			$\sigma$ & 0.1 & volatility\\
			$Q_0$ & 1 & quantity of shares \\
			$Y_0$ & 1 & initial target value \\
			$k$ & 0.95($Y_0-0.05$) & lower boundary \\
			$h$ & 1.05($Y_0+0.05$) & upper boundary \\
			\hline
		\end{tabular}
	\end{table}
	
	In Figure \ref{fig:dynamicsofJ} we plot the value function as a function of $y$ for different values of $\lambda$. The function increases in $\lambda$ and, therefore, decreases both in the permanent and temporary price impact of trades ($l$ and $b$), in the volatility of the asset price ($\sigma$) and increases in the slippage cost ($\gamma$).
	The interpretation is that impact costs (both permanent and temporary) render costly trading in the market and, therefore, the reward function is penalized. Similarly, as volatility increases, the selling strategy is not affected, but the downside risk of reaching the lower absorbing barrier increases rendering a lower value function. Finally, an increase in slippage cost induces the broker to quickly decrease inventory with a positive effect on performance.

	\begin{figure}[H]
		\centering
		\includegraphics[width=0.5\linewidth]{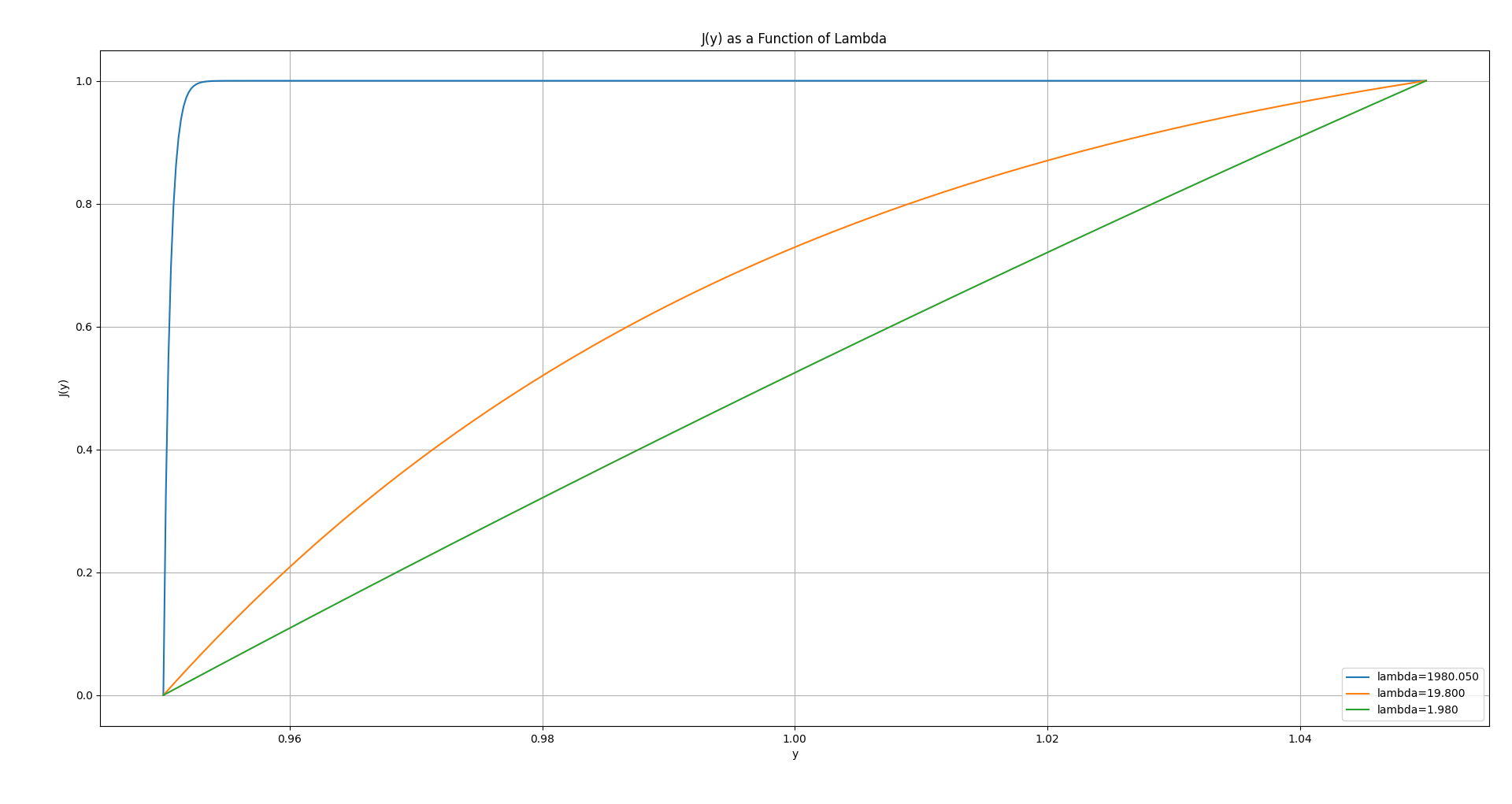}
		\caption{Value function for $\lambda=1.98, \ 19.80, \ 1980.05.$}
		\label{fig:dynamicsofJ}
	\end{figure}
	
	We evaluate the optimal strategy as the parameters $b$, $l$, $\gamma$ and $\sigma$ change. In Figure \ref{fig:dynamicofvtandYt_1}
	we reproduce $v^*(t), \ Q^*(t)$ and $Y^*(t)$ as functions of $t$ for different values of the parameters; on the third column we plot the average values obtained for $10,000$ simulations, the shaded regions
	represent the $5\%-95\%$ confidence interval of the simulations.
	
	As predicted by optimal solutions (\ref{OPT})-(\ref{OPT1}), $v^*(t)$ and $Q^*(t)$ monotonically decrease over time at an exponential rate. As far as the performance $Y^*(t)$ is concerned, we observe that it increases on average and quickly stabilizes at a value which can be far away from the upper barrier marking the success of the selling strategy ($1.05$). For example, for $\gamma=0.05$, the $5\%-95\%$ band does not include the $1.05$ level, i.e., the probability of reaching the upper barrier is smaller than 5\%. The rationale for this result lies on the fact that the broker quickly liquidates the shares because of the quadratic inventory penalization, this trading strategy reduces inventory costs but also leads to high execution costs which reduce wealth permanently. Quickly the trader ends with almost no shares, from then on she is not able to increase her wealth and, therefore, the performance (on average) stabilizes without reaching the upper barrier.
	Note that the average value of $Y^*(t)$, as well as the $5\%-95\%$ band, are well above the initial performance value and the lower barrier.
	\begin{figure}
		\centering
		\begin{subfigure}{1\textwidth}
			\centering
			\includegraphics[width=1\textwidth, height=0.5\textheight, keepaspectratio]{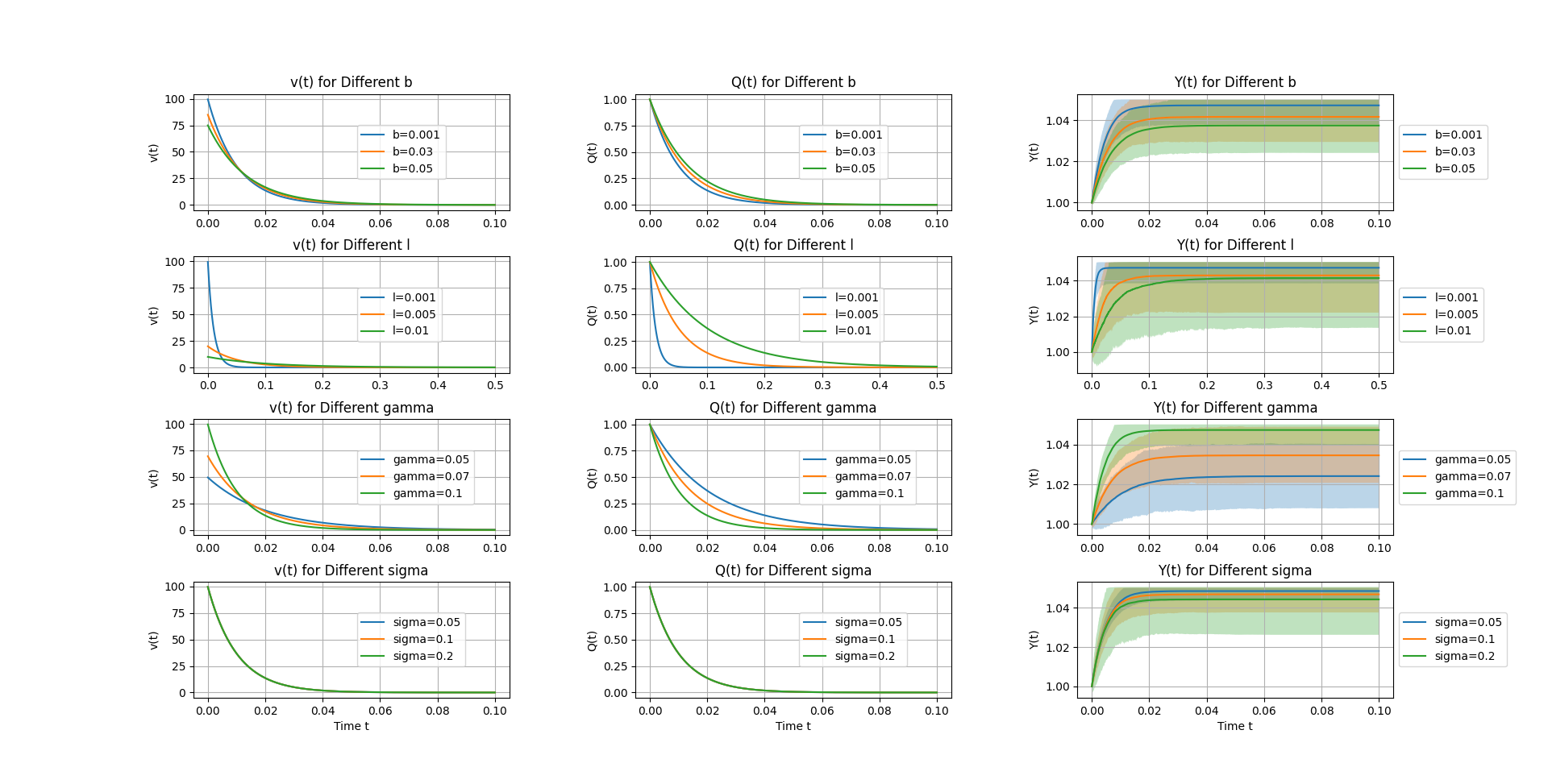}
			\caption{$v^*(t)$, $Q^*(t)$ and $Y^*(t)$ varying parameters $b, \ l, \ \gamma, \sigma$. 
				On the right we report the average values of the performance obtained for $10,000$ simulations, the shaded region
				represents the 5-95\% confidence intervals of the simulations.}
			\label{fig:dynamicofvtandYt_1}
		\end{subfigure}
		
		\vspace{0.5cm}
		\begin{subfigure}{1\textwidth}
			\centering
			\includegraphics[width=1\textwidth, height=0.4\textheight]{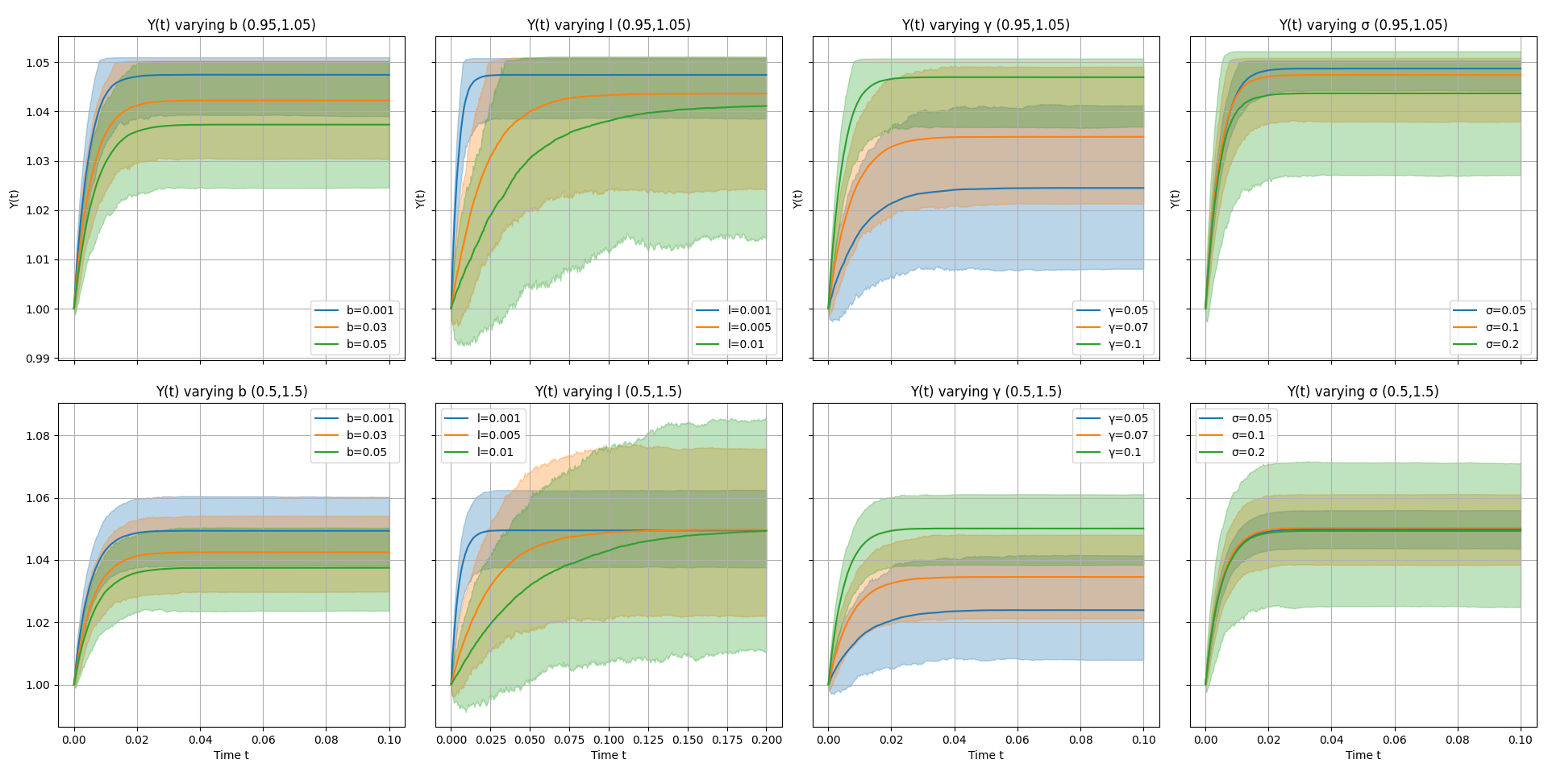}
			\caption{Average values of $Y^*(t)$ for $10,000$ simulations varying the parameters $b, \ l, \ \gamma, \sigma$; the shaded region represents the 5-95\% confidence intervals of the simulation results. The upper pictures concern $Y_{\min} = 0.95, Y_{\max} = 1.05$, the lower pictures concern $Y_{\min} = 0.5, Y_{\max} = 1.5$.}
			\label{fig:dynamicofYt_biggerband}
		\end{subfigure}
		
		\caption{Optimal strategies for Problem P1 varying the parameters $b, l, \gamma, \sigma$ and the barriers.}
		\label{fig:dynamicofvtandYt_combined}
	\end{figure}

	Looking at the optimal solution (\ref{OPT}), we observe that it doesn't depend on the asset price volatility and on the asset price itself. This is confirmed by the evolution of $v^*(t)$ and $Q^*(t)$, that are purely deterministic functions of time and do not depend on $\sigma$. This is different from what is observed for the optimal solution with an upper price limiter barrier (barrier only on one side for the asset price), see \cite{JAKIN}. In that case, the trading frequency increases as volatility goes up. Instead, considering symmetric barriers on the performance of the broker, the asset price volatility has non effect on the trading strategy.
	In the long run, the average value of $Y^*(t)$ slightly decreases and its dispersion increases in the asset price volatility. The rationale of these results is that the average performance is well above the initial value, a higher asset price volatility (with the same trading strategy) leads to higher probability on the two tails of the distribution of the performance but, while the upper absorbing barrier is not far away and the performance cannot overcome it, the lower barrier is far away. Therefore, there is almost the downside risk and, as the asset price volatility goes up,
	the performance on average decreases and its dispersion increases.
	
	As $\gamma$ increases, we observe 
	faster liquidation of the shares because the broker faces a higher quadratic inventory cost. 
	The trading strategy becomes more aggressive, this 
	negatively and linearly impacts the execution price and (permanently) $X(t)$ yielding a lower inventory quadratic cost. 
	The average performance goes up because the reduction of quadratic inventory costs prevails over the linear slippage term and the dispersion decreases because the broker quickly liquidates the shares. 
	%However, over a longer period we observe a lower average performance because the shares have been liquidated at poor conditions and the performance dispersion goes up because the broker doesn't hedge market risk intertemporally.
	
	An increase of the permanent effect associated with the broker's trades ($b$) has little impact on the liquidation rate, it lowers the long run average performance and increases the dispersion around the average. The rationale of this result is that as $b$ increases, the liquidation rate slightly decreases and the market price is significantly impacted. Poor execution conditions decrease the average value of the broker's performance. On the other hand, its dispersion goes up because the broker holds the shares for a longer period.
	
	A similar effect is observed for an increase of the temporary effect ($l$). The main difference with respect to $b$ is that we observe a less pronounced effect on the long run average performance.
	
	In Figure \ref{fig:dynamicofYt_biggerband}, we show the simulations of the broker's performance as the two barriers change, from $1\pm 0.05$ to $1\pm 0.5$. We do not report the time evolution of the trading strategy and of the inventory as they are not affected by the change. As the barriers widen, both the average value of the performance and its variance increase.
	
	In Figure \ref{fig:dynamicofvtandYt_withP0} we compare $v^*(t)$, $Q^*(t)$, and $Y^*(t)$ for Problem P1 with the solutions obtained for Problem P0. It turns out that the strategy obtained with a performance target liquidates the shares faster than the strategy obtained for Problem P0. The inventory obtained for Problem P1 exponentially decreases over time and is always below the one associated with Problem P0, which linearly decreases over time. 
	Anticipating the liquidation, the broker 
	faces higher linear execution costs but also rapidly decreases quadratic inventory costs. The first effect impacts $X(t)$ and therefore is permanent, the second one impacts the performance instantaneously.
	This explains why the performance for Problem P0 on average is lower than that for Problem P1 over a short time horizon and becomes higher over a longer horizon, see also Table \ref{tableP0P1} reporting the mean and the variance of the performance at $t=0.02, \ 0.06, \ 0.1$ for $T=1$.
	The higher average value in the long run for Problem P0 is 
	also due to the fact that there is no upper bound to the performance  in this problem and, therefore, it can overcome the upper barrier imposed in Problem P1.
	
	The dispersion of the performance for Problem P1 is smaller than for Problem P0 over a short horizon because the broker holds a smaller inventory and, therefore, less volatility in her book. A result that confirms what has been obtained in \cite{JAKIN} considering the optimal acquisition problem with a price limiter. Over a longer horizon, the dispersion of the performance for Problem P0 becomes smaller than that for Problem P1. The phenomenon is due to the fact that, liquidating the shares quickly in Problem P1, the broker cannot manage intertemporally the inventory because it goes quickly to zero,
	as a consequence the dispersion enlarges as time goes and then remains constant. Instead, managing the inventory over a fixed horizon, the broker can handle the volatility intertemporally and does a god job in the long run. 
	
	We can conclude that the relative performance of the two problems depends on the evaluation horizon. Over a short horizon, a performance (infinite horizon) target strategy dominates the classical finite horizon execution problem (higher mean and lower variance), over a longer period the reverse is observed.
	
	\begin{figure}
		\centering
		
		\begin{subfigure}{1\textwidth}
			\centering
			\includegraphics[width=1\linewidth]{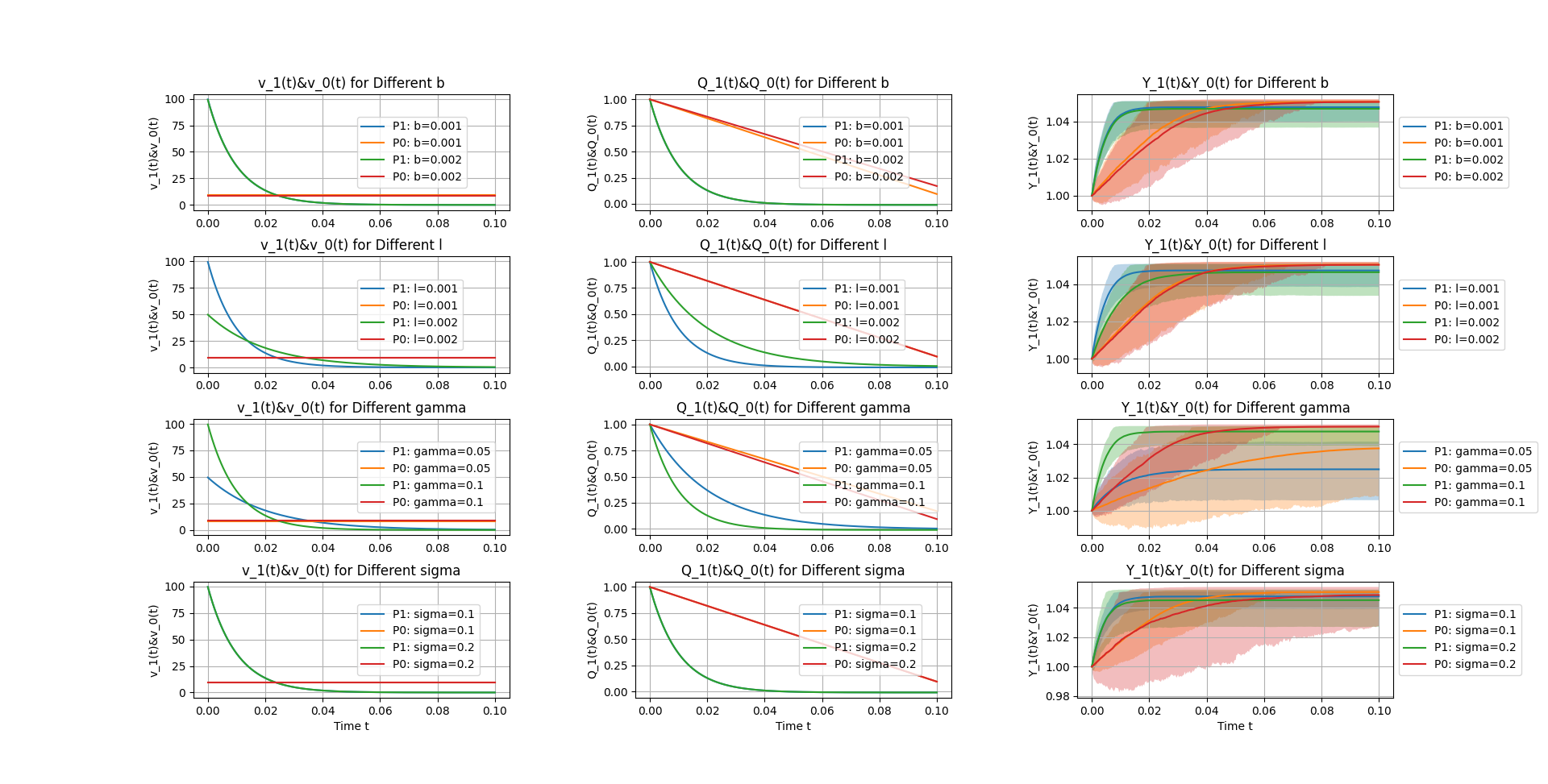}
			\caption{Optimal strategies obtained for Problem P1 ($v_1(t)$, $Q_1(t)$, and $Y_1(t)$), and for Problem P0 ($v_0(t)$, $Q_0(t)$, and $Y_0(t)$) varying the parameters $b, l, \gamma, \sigma$. 
				On the right we report the average of the performance obtained for $10,000$ simulations, the shaded regions
				represent the 5\%-95\% confidence interval of the simulations.}
			\label{fig:dynamicofvtandYt_withP0}
		\end{subfigure}

		\vspace{0.1cm}
		
		\begin{subfigure}{1\textwidth}
			\centering
			\includegraphics[width=0.8\linewidth]{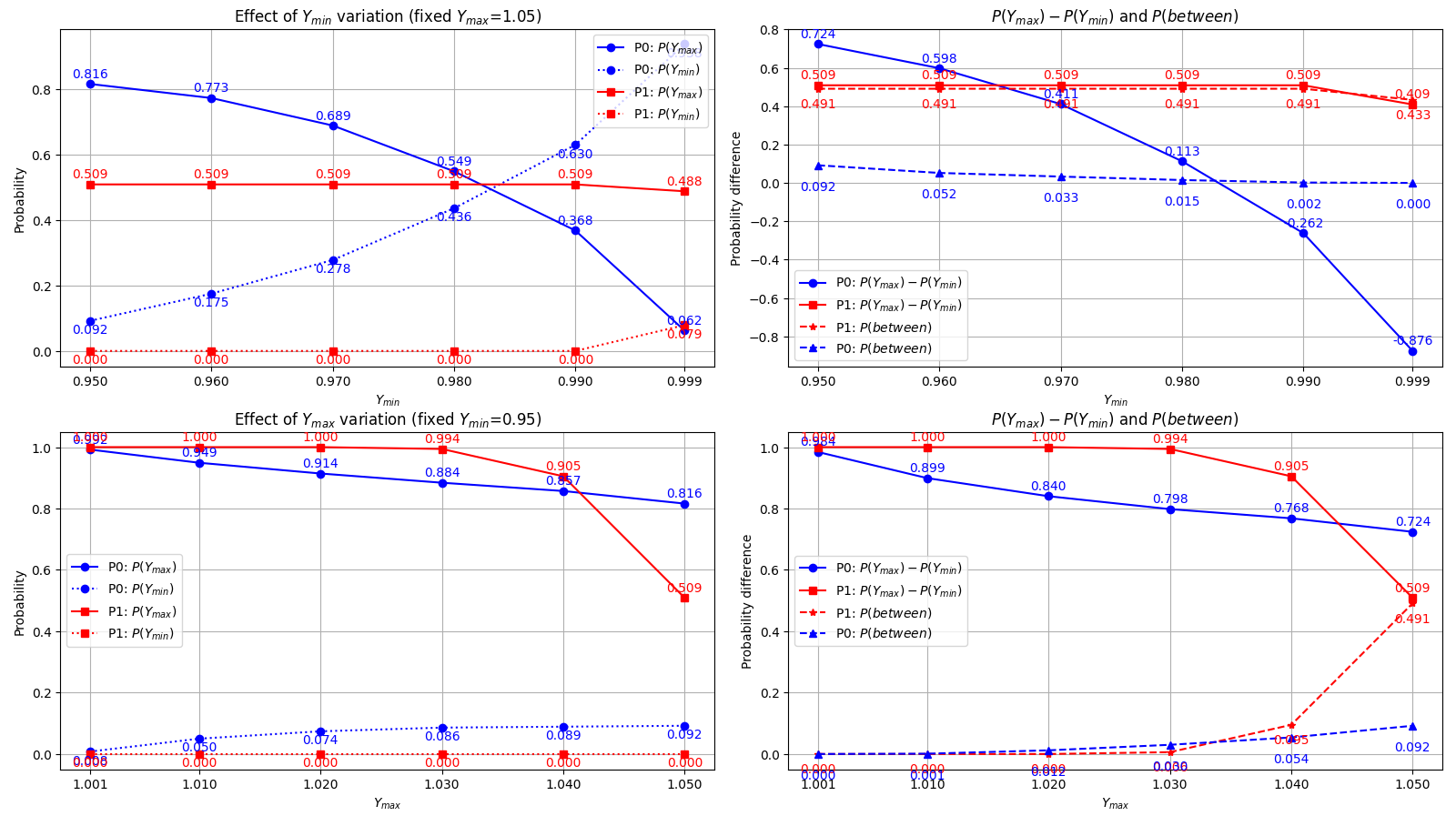}
			\caption{Probabilities of reaching $Y_{\max}$ or $Y_{\min}$ or neither by $t=T=1$ for 
				Problem P0 and P1, based on $10,000$ simulations. On the two plots on the left (top and bottom left),  blue solid lines represent the probability of hitting $Y_{\max}$ for Problem P0 ($P(Y_{max})$), while dotted lines represent the probability of hitting $Y_{\min}$ ($P(Y_{min})$); the same lines in red are for Problem P1. The solid lines on the right illustrate the probability difference $P(Y_{\max}) - P(Y_{\min})$, while the dotted lines represent the probability of remaining between the two bands at $t=T=1$ ($P(between$). 
			}
			\label{fig:probabilityp0p1}
		\end{subfigure}
		\caption{Simulation results for Problem P0 and P1.}
		\label{fig:Finalaveragevalue_Yt}
	\end{figure}

	We can evaluate the probability of reaching one of 
	the two barriers $Y_{\max}$ and $Y_{\min}$ first according to the solution for Problem P1 and to the one obtained for Problem P0. 
	We perform 10,000 simulations. Problem P0 has a terminal time $T$, while Problem P1 doesn't include a time horizon. In Figure \ref{fig:probabilityp0p1}
	we evaluate the probabilities for the two problems  
	of reaching $Y_{\max}, \ Y_{min}$ first or neither of the two by $t=T=1$.
	
	We emphasize that Problem P1 aims to maximize the probability of reaching $Y_{\max}$ before $Y_{\min}$, rather than simply maximizing the probability of reaching $Y_{\max}$. The broker has to balance the two goals.
	
	The probability of reaching the upper barrier for Problem P1 depends on the levels of the two barriers. For the original levels (0.95 and 1.05), the probability of reaching $Y_{\max}$ by time $t=1$ solving Problem P1 is $50.9\%$, the probability of reaching $Y_{\min}$ is zero and the probability of not reaching 
	$Y_{\max}$ and $Y_{\min}$ by $t=1$ is $49.1\%$. Instead, for Problem P0 we have a significant probability of reaching $Y_{\max}$ and $Y_{\min}$ by $t=1$, $81.6\%$ and $9.2\%$, respectively, and the probability of not reaching the two barriers by $t=1$ is $9.2\%$. It turns our that Problem P0 leads to the two extremes while Problem P1 leads to the upper barrier or to remain between the two barriers. 
	
	As the lower and upper barriers change (first and second set of pictures, respectively), we observe different phenomena. 
	
	As far as Problem P1 is concerned, when the lower barrier changes, the probability of success as well as the probability of not touching the two barriers by $t=1$ remain almost constant at $50.9\%$ and $49.1\%$, respectively, with null probability of touching the lower barrier; only in case of a lower barrier very close to the initial performance value ($Y_{\min}=0.999$) we observe a positive (small) probability of touching it by $t=1$.

	\begin{table}[htbp]
		\centering
		\caption{Mean and variance of the performance for Problem P0 and P1 as depicted in Figure \ref{fig:dynamicofvtandYt_withP0} at $t=0.02, \ 0.06, \ 0.10$ based on $10,000$ simulations, $T=1$.}
		\begin{tabular}{ccccccc}
			\toprule
			Variation & Parameter & Time & P1 Mean & P1 Variance & P0 Mean & P0 Variance \\
			\midrule
			\multirow{10}{*}{$b$} 
			& \multirow{5}{*}{0.001} 
			& 0.02 & 1.04709 & 0.00002 & 1.03149 & 0.00015 \\
			&                           & 0.06 & 1.04742 & 0.00002 & 1.04996 & 0.00001 \\
			&                           & 0.10 & 1.04741 & 0.00002 & 1.05062 & 0.00000 \\
			\cmidrule{2-7}
			& \multirow{5}{*}{0.002} 
			& 0.02 & 1.04680 & 0.00002 & 1.02831 & 0.00016 \\
			&                           & 0.06 & 1.04719 & 0.00002 & 1.04935 & 0.00002 \\
			&                           & 0.10 & 1.04719 & 0.00002 & 1.05051 & 0.00000 \\
			\midrule
			\multirow{10}{*}{$l$} 
			& \multirow{5}{*}{0.001} 
			& 0.02 & 1.04709 & 0.00002 & 1.03123 & 0.00014 \\
			&                           & 0.06 & 1.04742 & 0.00002 & 1.05005 & 0.00001 \\
			&                           & 0.10 & 1.04743 & 0.00002 & 1.05061 & 0.00000 \\
			\cmidrule{2-7}
			& \multirow{5}{*}{0.002} 
			& 0.02 & 1.04269 & 0.00006 & 1.03028 & 0.00015 \\
			&                           & 0.06 & 1.04666 & 0.00003 & 1.04977 & 0.00001 \\
			&                           & 0.10 & 1.04670 & 0.00003 & 1.05050 & 0.00000 \\
			\midrule
			\multirow{10}{*}{$\gamma$} 
			& \multirow{5}{*}{0.05} 
			& 0.02 & 1.02154 & 0.00009 & 1.01405 & 0.00017 \\
			&                           & 0.06 & 1.02489 & 0.00010 & 1.03207 & 0.00024 \\
			&                           & 0.10 & 1.02496 & 0.00010 & 1.03795 & 0.00019 \\
			\cmidrule{2-7}
			& \multirow{5}{*}{0.1} 
			& 0.02 & 1.04654 & 0.00002 & 1.02948 & 0.00015 \\
			&                           & 0.06 & 1.04691 & 0.00002 & 1.04962 & 0.00002 \\
			&                           & 0.10 & 1.04691 & 0.00002 & 1.05056 & 0.00000 \\
			\midrule
			\multirow{10}{*}{$\sigma$} 
			& \multirow{5}{*}{0.1} 
			& 0.02 & 1.04685 & 0.00002 & 1.03037 & 0.00015 \\
			&                           & 0.06 & 1.04719 & 0.00002 & 1.04986 & 0.00001 \\
			&                           & 0.10 & 1.04719 & 0.00002 & 1.05063 & 0.00000 \\
			\cmidrule{2-7}
			& \multirow{5}{*}{0.2} 
			& 0.02 & 1.04456 & 0.00008 & 1.02935 & 0.00045 \\
			&                           & 0.06 & 1.04496 & 0.00008 & 1.04622 & 0.00021 \\
			&                           & 0.10 & 1.04496 & 0.00008 & 1.04865 & 0.00012 \\
			\bottomrule
		\end{tabular}
		\label{tableP0P1}
	\end{table}

	When the upper barrier changes, the probability of touching it remains constant next to 100\%, only for a high threshold ($1.05$) the probability significantly decreases. When this happens, 
	we do not observe an increase of the probability of reaching the lower barrier, instead we observe a significant probability for the performance to remain between the two barriers by $t=1$.

	As far as Problem P0 is concerned, we observe that the probability of hitting the upper (lower) barrier by time $t=1$ goes down (up) with $Y_{\max}$ and $Y_{\min}$. For a high $Y_{\min}$ the probability of touching the lower barrier is higher than that of touching the upper barrier.
	
	Interesting enough, while the sum of the probabilities of touching $Y_{\max}$ and $Y_{\min}$ is next to $1$ for Problem P0 with a very low probability of terminating between the two barriers, this is not the case for Problem P1. In case the probability of hitting first $Y_{\max}$ is less than $1$, the P1 strategy assigns a positive probability to lie between the two barriers by $t=1$ with a very low probability of touching the lower barrier.
	
	The key insight is that the P1 strategy is aggressive ($Q^*(t)$ exponentially decreases) and quickly goes above $Y_0$, but then almost only two events can occur: the performance touches the upper barrier or ends without touching both barriers because the inventory is null.
	Instead, the P0 strategy is less aggressive ($Q^*(t)$ linearly decreases), this almost leads to three feasible events: touching the upper, the lower barrier or neither of the two barriers. 
	The P1 trading strategy is at the same time more aggressive (it quickly goes above the initial performance) and cautious, minimizing the probability of touching the lower barrier. Instead, the classical solution takes risks on both sides.
	The P1 strategy seems to sacrifice the potential to achieve $Y_{\max}$ 
	minimizing the downside risk of reaching $Y_{\min}$.
	
	We can conclude that optimal execution with a target performance does not take excessive risk but is more aggressive and this leads to liquidation of the shares under poor conditions with a poor performance over a long horizon. The strategy is not affected by excessive risk taking but by shortermism. 
	
	\section{Comparison with a price limiter strategy}
	\label{PRIC}
	As a performance criterion we set 
	\begin{equation}
		Y(t) = X(t) + Q(t)(S(t) - \gamma Q(t)) - \phi \int_0^{t} Q(u)^2 \,du 
		\label{ACEQ}
	\end{equation}
	with $Y(0) = X(0) + Q(0)(S(0) - \gamma Q(0))$.
	The problem is to find the strategy that maximizes (\ref{MAX}), and we refer to it as Problem $P1^{\prime}$.
	The optimal strategy is obtained in the following proposition.
	
	\begin{proposition}
		The optimal strategy for Problem $P1^{\prime}$ is
		\begin{equation}
			\label{OP1}
			v^{*}(t)= \frac{2\gamma-b}{2l}Q^*(t), \ \ 
			v^*(t)=\frac{2\gamma-b}{2l}e^{\frac{b-2\gamma}{2l}t}Q_0
		\end{equation}   
		\begin{equation}  
			\label{OP2}
			Q^{*}(t)=e^{\frac{b-2\gamma}{2l}t}Q_0.
		\end{equation}
		The value function is 
		\begin{equation}
			\mathcal{J}(y)=\frac{e^{-\lambda y}-e^{-\lambda k}}{e^{-\lambda h}-e^{-\lambda k}},
		\end{equation}
		where
		$$\lambda=\frac{(-b+2\gamma )^2-4l\phi}{2l\sigma^2}.$$
	\end{proposition}
	\begin{proof}
		The value function $\mathcal{J}(y, q)$ satisfies the second order HJB equation
		\begin{align}
			\left\{
			\begin{aligned}
				&\sup_{v\in\mathcal{A}}\mathcal{L}^{v}\mathcal{J}(y, q)=0, &&(y, q)\in\Omega, \\
				&\mathcal{J}(y, q)=\Phi(y, q), &&(y, q)\in\partial^{*}\Omega,
				\label{HJB-linear case-Target ax+bqs-gammaq2}
			\end{aligned}
			\right.
		\end{align}
		where $\mathcal{L}^{v}$ is the infinitesimal generator operator
		\begin{align*}
			\mathcal{L}^{v}\mathcal{J}(y, q)=\frac{1}{2}&\sigma^2q^2\partial_{yy}\mathcal{J}--\phi q^2\partial_y\mathcal{J}
			+\{[- lv^2- qbv+2\gamma qv(t)]\partial_{y}\mathcal{J}-v\partial_{q}\mathcal{J}\}
		\end{align*}
		and
		\[
		\Phi(y, q) = \left\{
		\begin{aligned}
			&1, &&y=h, \\
			&0, &&y=k.
		\end{aligned}
		\right.
		\]
		The supremum is reached at 
		$$
		v^*(y, q)=\frac{-b q\partial_{y}\mathcal{J}+2\gamma q\partial_{y}\mathcal{J}-\partial_{q}\mathcal{J}}{2 l\partial_{y}\mathcal{J}}
		$$
		and the value function $\mathcal{J}^{v}(y, q)$ satisfies 
		\begin{align}
			\left\{
			\begin{aligned}
				&\frac{1}{2}\sigma^2q^2\partial_{yy}\mathcal{J}-\phi q^2\partial_y\mathcal{J}+\frac{(-b q\partial_{y}\mathcal{J}+2\gamma q\partial_{y}\mathcal{J}-\partial_{q}\mathcal{J})^2}{4l\partial_{y}\mathcal{J}}=0, &&(y, q)\in\Omega, \\
				&1, &&y= h, \\
				&0, &&y= k.
			\end{aligned}
			\right.
		\end{align}
		Observing the boundary conditions, we assume that the value function depends only on $y$. The HJB equation becomes:
		\begin{align*}
			&\frac{1}{2}\sigma^2q^2\partial_{yy}\mathcal{J}-\phi q^2\partial_y\mathcal{J}+\frac{(-b q\partial_{y}\mathcal{J}+2\gamma q\partial_{y}\mathcal{J})^2}{4l\partial_{y}\mathcal{J}}=0.
		\end{align*}
		Reorganizing the equation, we have a linear second-order ordinary differential equation:
		\begin{align*}
			&\frac{1}{2}\sigma^2\partial_{yy}\mathcal{J}+\frac{[(-b+2\gamma )^2-4l\phi]\partial_{y}\mathcal{J}}{4l}=0,
		\end{align*}
		with boundary conditions
		$$
		J(h)=1,J(k)=0.
		$$
		The remaining computations are as in Proposition \ref{Prop1}, with $\lambda=\frac{(-b+2\gamma )^2-4l\phi}{2l\sigma^2}$.
	\end{proof}

	In \cite{JAKIN}, the optimal execution problem for the acquisition of shares with a price limiter is considered, in what follows we reformulate it for the problem of selling shares. The objective function is
	$$
	E[\int_0^{\tau}(S(u)-lv(u))v(u)du + Q(\tau)(S(\tau) - \gamma Q(\tau)) - \phi \int_0^{\tau} Q(u)^2 \, du], 
	$$
	where $\tau=T\wedge\{t:Q(t)=0\}\wedge\inf\{t:S(t)=\underline{S}\}$ is the first time that either $S(t)$ reaches the lower limit price $\underline{S}$, there are no more shares to be liquidated, or the terminal time $T$ is reached.
	As in \cite{JAKIN} we set $b=0$ to reduce the dimension of the problem.
	
	We consider the classical Almgren-Chriss (AC) strategy that aims to maximize the expected value of (\ref{ACEQ}) over time $T$, see \cite{CAR15}:
	
	\begin{equation*}
		v^*(t) = \Gamma \frac{\zeta e^{\Gamma (T-t)} + e^{-\Gamma (T-t)}}{\zeta e^{\Gamma (T-t)} - e^{-\Gamma (T-t)}} Q^{*}(t),
	\end{equation*}
	\begin{equation*}
		Q^{*}(t) = \frac{\zeta e^{\Gamma (T-t)} - e^{-\Gamma (T-t)}}{\zeta e^{\Gamma T} - e^{-\Gamma T}} Q_0,
	\end{equation*}
	where
	\begin{equation*}
		\Gamma = \sqrt{\frac{\phi}{l}}, \quad
		\zeta = \frac{\gamma - \frac{1}{2} b + \sqrt{l\phi}}{\gamma - \frac{1}{2} b - \sqrt{l\phi}}.
	\end{equation*}
	
	We now provide a numerical analysis of the solution of Problem P1' comparing it with the AC strategy and the one in \cite{JAKIN}.
	The set of parameters is reported in Table \ref{parame}.
	
	\begin{table}[htbp]
		\centering
		\caption{Set of parameters of the model.}
		\label{parame}
		\begin{tabular}{|c|c|c|}
			\hline
			Parameter & Value & Interpretation \\
			\hline
			$l$ & 0.0001 & temporary impact \\
			$\gamma$ & 0.1 & slippage cost \\
			$\sigma$ & 0.1 & volatility \\
			$\phi$ & 0.001 & running penalty\\
			$X_0$ & 0 & initial wealth function\\
			$Q_0$ & 1 & quantity of shares\\
			$S_0$ & 20 & initial target value \\
			$\underline{S}$ & 19.9 & lower limit of price\\
			%   $\overline{S}$ & 20.5 & upper limit of price\\
			$k$ & 19.85($Y_0-0.05$) & lower boundary \\
			$h$ & 19.95($Y_0+0.05$) & upper boundary \\
			$T$ & 1 & terminal time\\
			\hline
		\end{tabular}
	\end{table}
	
	In Figure \ref{fig:Compare_PLimit_Q}, we compare the three strategies. The strategy with a price limiter is quite similar to the classical one; instead, the one obtained with two boundaries is much more aggressive.
	
	In Figure \ref{fig:enter-label} we show the histogram of the performance for each strategy: (Blue) the price limiter strategy, (Red) the AC strategy, and (Green) the target strategy. Notice that the performance obtained for Problem P1' is much more concentrated with respect to the classical one and the one obtained with a price limiter.

	\begin{figure}[H]
		\centering
		\includegraphics[width=0.8\linewidth]{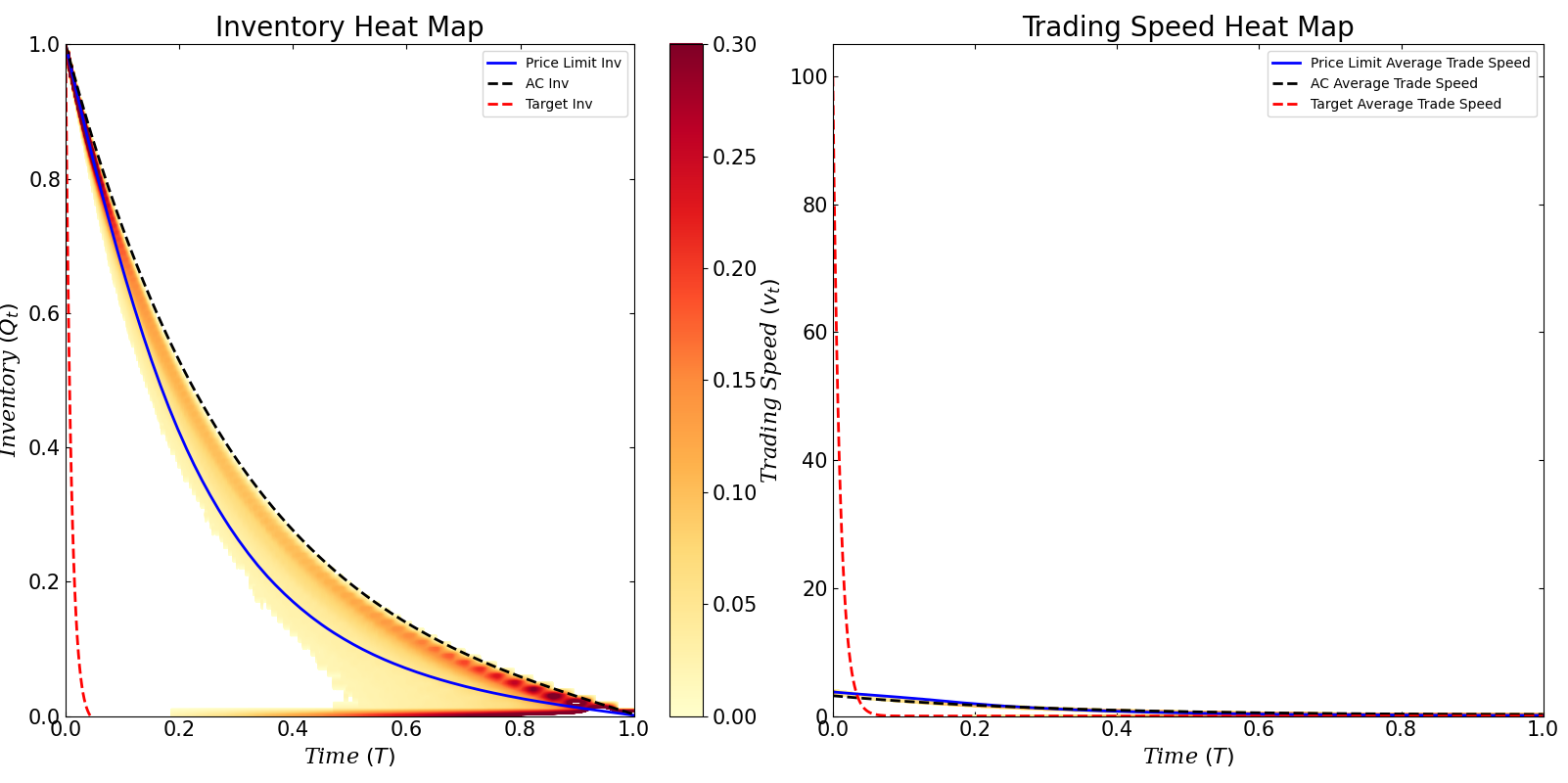}
		\caption{The left hand figure shows $Q^*(t)$ in (\ref{OP2}), compared with the average trajectory of price limiter optimal strategy and the AC optimal strategy. The right hand figure shows $v^*(t)$ as given in (\ref{OP1}) compared to the other strategies. The color bar in the left figure indicates the density of the simulation for the price limiter optimal strategy.}
		\label{fig:Compare_PLimit_Q}
	\end{figure}
	
	\begin{figure}[H]
		\centering
		\includegraphics[width=0.7\linewidth]{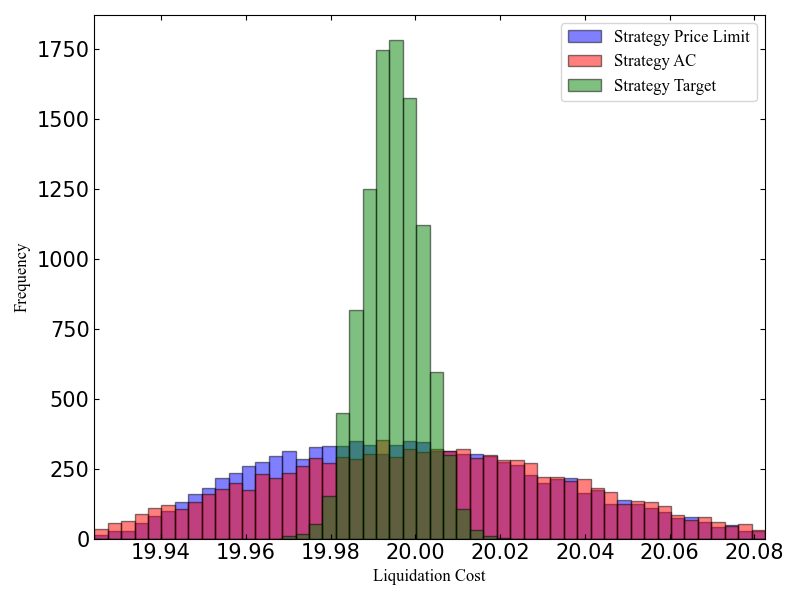}
		\caption{Histogram of the liquidation  cost, (Blue) Limit price strategy, (Red) AC strategy and (Green) target strategy.}
		\label{fig:enter-label}
	\end{figure}
	
	\section{Conclusions}
	\label{CONC}
	The liquidation of a large number of shares leads to the optimal execution problem that can be embedded in a principal-agent relationship between the owner of the shares and the broker who has to execute the order. The remuneration of the broker should be designed to induce him to act in the interest of the owner of the shares avoiding excessive risk taking.
	
	In this paper, we have considered the case of a fixed remuneration if the performance of the broker (made up of cash and a penalization for the inventory) reaches an upper barrier and zero remuneration in case a lower barrier is reached.
	
	A remuneration mechanism related to reaching a performance target is considered a bad contract in the executive compensation literature, as it may induce excessive risk taking and shortermism. In our analysis, we have shown that a fixed remuneration in case of success and a null one in case of poor performance leads to a shortermism bias but not excessive risk taking.
	
	The optimal strategy with a performance target foresees a liquidation of the shares at a much higher rate compared to the solution obtained according to the classical execution problem. The strategy obtained with a performance target yields a higher performance and a smaller dispersion compared to the solution with a finite horizon over a short horizon, but over a long horizon the reverse is observed.
	
	The strategy turns out to be aggressive and conservative at the same time. Shares are quickly liquidated, the performance goes up, but then the broker balances the goal of reaching the upper barrier avoiding the lower one. As a consequence, there is a high probability of reaching the upper barrier and of remaining between the two barriers with a very low probability of reaching the lower performance barrier. Instead, the strategy obtained by solving the classical problem allocates a significant probability to reach both extremes.


\begin{thebibliography}{99}
		\bibitem[Almgren and Chriss, 2001]{ALGCH}
		Almgren, R. and Chriss, N. (2001)
		Optimal execution of portfolio transactions, {\em The Journal of Risk}, 3: 5-39.
		
		\bibitem[Baldauf et al., 2022]{BALD}
		M. Baldauf, C. Frei, J. Mollner (2022) Principal Trading Arrangements: When Are Common Contracts Optimal?, {\em Management Science}, 68(4), 3112–3128.
		
		\bibitem[Barucci et al., 2018]{BARMA2}
		Barucci, E., La Bua, G. and Marazzina, D. (2018)
		On relative performance, remuneration and risk taking of asset managers, {\em Annals of Finance}, 14(4): 517-545.
		
		\bibitem[Barucci and Marazzina, 2015]{BARMA}
		Barucci, E. and Marazzina, D. (2015)
		Risk seeking, nonconvex remuneration and regime switching, {\em International Journal of Theoretical and Applied Finance}, 18, 2.
		
		\bibitem[Basak et al., 2007]{BASA} 
		Shapiro, A. Basak, S., Pavlova, A., Shapiro, A.(2007). Optimal asset allocation and risk shifting in money management, {\em The Review of Financial Studies}, 20, 1583–1621.
		
		\bibitem[Bertsimas and Lo, 1998]{BERT}
		Bertsimas, D. and Lo, A. (1998) 
		Optimal control of execution costs,  {\em Journal of Financial Markets}, 1, 1-50.
		
		\bibitem[Browne, 1995]{BR95}
		Browne S. (1995) Optimal investment policies for a firm with a
		random risk process: exponential utility and minimizing the
		probability of ruin, {\em Mathematics of Operations Research}, 20, 4: 937-958.
		
		
		\bibitem[Browne, 1998]{BR98}
		Browne S. (1998) The return on investment from proportional
		portfolio strategies. {\em Advances Applied Probability}, 30: 216
		238.
		
		\bibitem[Browne, 1999]{BR99}
		Browne S. (1999) Reaching goals by a deadline: digital options
		and continuous time active portfolio management. {\em Advances in
			Applied Probability}, 31: 551-577.
		
		\bibitem[Carpenter, 2000]{CARP}
		J. Carpenter (2000) Does option compensation increase managerial risk appetite? {\em Journal of Finance}, 21, 2311–2331.
		
		\bibitem[Cartea et al., 2015]{CAR15}
		Cartea, Á., Jaimungal, S. and Penalva, J. (2015). Algorithmic and high-frequency trading. {\em Cambridge University Press}.
		
		\bibitem[Donnely, 2022]{DON}
		Donnelly, R (2022) Optimal Execution: A Review, {\em Applied Mathematical Finance}, 29:3, 181-212.
		
		
		\bibitem[Fleming and Rishel, 2012]{FLE}
		Fleming, W. H., and Rishel, R. W. (2012). Deterministic and stochastic optimal control (Vol. 1). {\em Springer Science \& Business Media}.
		
		\bibitem[Grinblatt and Titman, 1989]{GRIN}
		M. Grinblatt \& S. Titman (1989) Adverse risk incentives and the design of performance- based contracts, {\em Management Science}, 35, 807–822.
		
		\bibitem[Jaimungal and Kinzebulatov, 2014]{JAKIN}
		Jaimungal, S. and Kinzebulatov, D. (2014) Optimal Execution with a Price Limiter, RISK, July 2014.
		
		\bibitem[Larsson et al., 2025]{LARS}
		Larsson, M., Muhle-Karbe, J. and Weber, B. (2025)
		Optimal contracts for delegated order execution, {\em Mathematical Finance}, 1-17.
		
		\bibitem[Ross, 2004]{ROSS}
		S. Ross (2004) Compensation, incentives, and the duality of risk aversion and riskiness, {\em Journal of Finance}, 59, 207–225.
		
	\end{thebibliography}
\end{document}